\documentclass[a4paper,UKenglish,cleveref, autoref, thm-restate,authorcolumns]{lipics-v2019}
\usepackage{tikz}
\usepackage{graphbox}
\usetikzlibrary{automata,arrows,shapes,decorations.pathreplacing,patterns}
\tikzstyle{every node}=[circle,draw=black, inner sep=1.5pt,fill=white]

\newcommand{\ltdots}{..}

\title{Linear Time Construction of Indexable Founder Block Graphs}

\author{Veli M{\"a}kinen}{Department of Computer Science, University of Helsinki, Finland}{veli.makinen@helsinki.fi}{http://orcid.org/0000-0003-4454-1493}{}

\author{Bastien Cazaux}{Department of Computer Science, University of Helsinki, Finland}{}{https://orcid.org/0000-0002-1761-4354}{}

\author{Massimo Equi}{Department of Computer Science, University of Helsinki, Finland}{massimo.equi@helsinki.fi}{}{}

\author{Tuukka Norri}{Department of Computer Science, University of Helsinki, Finland}{tuukka.norri@helsinki.fi}{http://orcid.org/0000-0002-8276-0585}{}

\author{Alexandru I. Tomescu}{Department of Computer Science, University of Helsinki, Finland}{alexandru.tomescu@helsinki.fi}{https://orcid.org/0000-0002-5747-8350}{}

\authorrunning{V. M\"akinen \emph{et al}.}

\Copyright{Veli M\"akinen, Bastien Cazaux, Massimo Equi, Tuukka Norri, and Alexandru I. Tomescu}



\category{} 

\relatedversion{} 

\supplement{}

\nolinenumbers 

\hideLIPIcs  

\EventEditors{John Q. Open and Joan R. Access}
\EventNoEds{2}
\EventLongTitle{42nd Conference on Very Important Topics (CVIT 2016)}
\EventShortTitle{CVIT 2016}
\EventAcronym{CVIT}
\EventYear{2016}
\EventDate{December 24--27, 2016}
\EventLocation{Little Whinging, United Kingdom}
\EventLogo{}
\SeriesVolume{42}
\ArticleNo{23}

\begin{document}

\maketitle

\begin{abstract}
We introduce a compact pangenome representation based on an optimal segmentation concept that aims to reconstruct \emph{founder sequences} from a multiple sequence alignment (MSA). Such founder sequences have the feature that each row of the MSA is a recombination of the founders. Several linear time dynamic programming algorithms have been previously devised to optimize segmentations that induce \emph{founder blocks} that then can be concatenated into a set of founder sequences. All possible concatenation orders can be expressed as a \emph{founder block graph}. We observe a key property of such graphs: if the node labels (founder segments) do not repeat in the paths of the graph, such graphs can be indexed for efficient string matching. We call such graphs \emph{segment repeat-free founder block graphs}. 

We give a linear time algorithm to construct a segment repeat-free founder block graph given an MSA. The algorithm combines techniques from the founder segmentation algorithms (Cazaux et al. SPIRE 2019) and \emph{fully-functional bidirectional Burrows-Wheeler index} (Belazzougui and Cunial, CPM 2019). We derive a succinct index structure to support queries of arbitrary length in the paths of the graph.

Experiments on an MSA of SAR-CoV-2 strains are reported. An MSA of size $410\times 29811$ is compacted in one minute into a segment repeat-free founder block graph of 3900 nodes and 4440 edges. The maximum length and total length of node labels is 12 and 34968, respectively. The index on the graph takes only $3\%$ of the size of the MSA. 
\end{abstract}

\section{Introduction}

Computational pangenomics \cite{marschall2016computational} ponders around the problem of expressing a reference genome of a species in a more meaningful way than as a string of symbols. The basic problem in such generalized representations is that one should still be able to support string matching type of operations on the content. 
Another problem is that any representation generalizing set of sequences also expresses sequences that may not be part of the real pangenome. That is, a good representation should have a feature to control over-expressiveness and simultaneously support efficient queries.

In this paper, we develop the theory around one promising pangenome representation candidate, the \emph{founder block graph}. This graph is a natural derivative of segmentation algorithms \cite{NCKM19,CKMN19} related to \emph{founder sequences} \cite{Ukkonen02}. 

Consider a set of individuals represented as lists of variations from a common reference genome. Such a set can be expressed as a \emph{variation graph} or as a \emph{multiple sequence alignment}. The former expresses reference as a backbone of an automaton, and adds a subpath for each variant. The latter inputs all variations of an individual to the reference, creating a row for each individual into a multiple alignment. Figure~\ref{fig:founder-alignment} shows an example of both structures with 6 very short genomes. 

A multiple alignment of much fewer \emph{founder} sequences can be used to approximate the input represented as a multiple alignment as well as possible, meaning that each original row can be mapped to the founder multiple alignment with a minimum amount of row changes (discontinuities). Finding an optimal set of founders is NP-hard \cite{RU07}, but one can solve relaxed problem statements in linear time \cite{NCKM19,CKMN19}, which are sufficient for our purposes. As an example on the usefulness of founders, Norri et al. \cite{NCKM19} showed that, on a large public dataset of haplotypes of human genome, the solution was able to replace 5009 haplotypes with only 130 founders so that the average distance between row jumps was over 9000 base pairs \cite{NCKM19}. This means that alignments of short reads (e.g. 100 bp) very rarely hit a discontinuity, and the space requirement drops from terabytes to just tens of gigabytes. Figure~\ref{fig:founder-alignment} shows such a solution on our toy example. 

\begin{figure}
    \centering
    \subcaptionbox{\label{fig:sub:sequences}}
    {\includegraphics[width=0.4\textwidth]{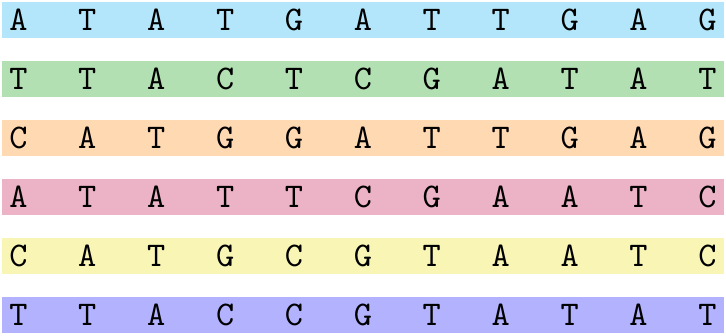}
    }
    \subcaptionbox{\label{fig:sub:founders}}
    {
    \raisebox{6mm}{\includegraphics[width=0.4\textwidth]{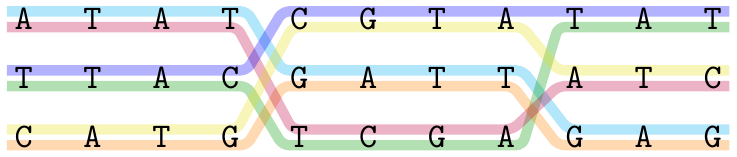}}
    }

    \subcaptionbox{\label{fig:sub:graph}}
    {
    \includegraphics[width=0.4\textwidth]{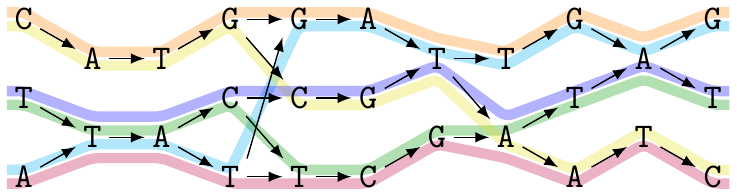}
    }
    \subcaptionbox{\label{fig:sub:founder_graph}}
    {
    \includegraphics[width=0.4\textwidth]{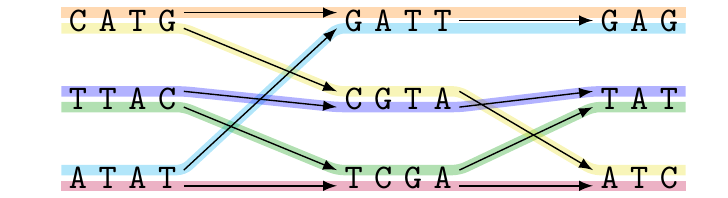}
    }
    \caption{(\protect\subref{fig:sub:sequences}) Input multiple alignment, (\protect\subref{fig:sub:founders}) a set of founders with common recombination positions as a solution to a relaxed version of founder reconstruction, (\protect\subref{fig:sub:graph}) a variation graph encoding the input and (\protect\subref{fig:sub:founder_graph}) a founder block graph.
    Here the input alignment and the resulting variation graph are unrealistically bad; the example is made to illustrate the founders.
    \label{fig:founder-alignment}}
\end{figure}

A \emph{block graph} is a labelled directed acyclic graph consisting of consecutive \emph{blocks}, where a block represents a set of sequences of the same length as parallel (unconnected) nodes. There are edges only from nodes of one block to the nodes of the next block. A \emph{founder block graph} is a block graph with blocks representing the segments of founder sequences corresponding to the optimal segmentation \cite{NCKM19}. Fig.~\ref{fig:founder-alignment} visualises such a founder block graph: There the founder set is divided into 3 blocks with the first, the second, and the third containing sequences of length 4, 4, and 3, respectively. The coloured connections between sequences in consecutive blocks define the edges. Such graphs interpreted as automata recognise the input sequences just like variation graphs, but otherwise recognise a much smaller subset of the language. With different optimisation criteria to compute the founder blocks, one can control the expressiveness of this pangenome representation. 

In this paper, we show that there is a natural subclass of founder block graphs that admit efficient index structures to be built that support exact string matching on the paths of such graphs. Moreover, we give a linear time algorithm to construct such founder block graphs from a given multiple alignment. The construction algorithm can also be adjusted to produce a subclass of \emph{elastic degeneralized strings} \cite{bernardini_et_al2019elastic}, which also support efficient indexing.  

The founder block graph definition given above only makes sense if we assume that our input multiple alignment is \emph{gapless}, meaning that the alignment is simply produced by putting strings of equal length under each other, like in Figure~\ref{fig:founder-alignment}. We develop the theory around founder block graphs under gapless multiple alignments. However, most of the results can be easily extended to handle gaps properly.

We start in Sect.~\ref{sect:related} by putting the work into the context of related work. In Sect.~\ref{sect:definitions} we introduce the basic notions and tools. In Sect.~\ref{sect:repeat-freeness} we study the property of founder block graphs that enable indexing. In Sect.~\ref{sect:construction} we give the linear time construction algorithm. In Sect.~\ref{sect:compressedindexing} we develop a succinct index structure that supports exact string matching in linear time. 
In Sect.~\ref{sect:gaps} we consider the general case of having gap symbols in multiple alignment. We report some preliminary experiments in Sect.~\ref{sect:experiments} on the construction and indexing of founder block graphs for a collection of SARS-CoV-2 strains. We consider future directions in Sect.~\ref{sect:discussion}. 

\section{Related work \label{sect:related}}

Indexing directed acyclic graphs (DAGs) for exact string matching on its paths was first studied by Sir\'en et al. in WABI 2011 \cite{SVM14}. A generalization of Burrows-Wheeler transform \cite{BW94} was proposed that supported near-linear time queries. However, the proposed transformation can grow exponential size in the worst case. Many practical solutions have been proposed since then, that either limit the search to short queries or use more time on queries \cite{thachuk2012indexing,HPB13,Iqp16,Siren18,Garetal18}. More recently, such approaches have been captured by the theory on \emph{Wheeler graphs} \cite{GMS17,GT19,Alaetal20}. 

Since it is NP-hard to recognize if a given graph is Wheeler \cite{GT19}, it is of interest to look for other graph classes that could provide some indexability functionality. Unfortunately, quite simple graphs turn out to be hard to index \cite{EGMT19,EMT20} (under the Strong Exponential Time Hypothesis). In fact, the reductions by Equi et al.~\cite{EGMT19,EMT20} can be adjusted to show that block graphs cannot be indexed in polynomial time to support fast string matching. But as we will later see, further restrictions on block graphs change the situation.   

Block graphs have also tight connection to \emph{generalized degenerate} (GD) strings and their \emph{elastic} version. These can also be seen as DAGs with a very specific structure. Matching a GD string is computationally easier and even linear time online algorithms can be achieved to compare two such strings, as analyzed by Alzamel et al.~\cite{alzamel_et_al18degenerate}. The elastic counterpart requires more care, as studied by Bernardini et al.~\cite{bernardini_et_al2019elastic}. Our results on founder block graphs can be casted on GD strings and elastic strings, as we will show later.

Finally, our indexing solution has connections to succinct representations of \emph{de Bruijn graphs} \cite{BOSS12,Bouetal15,Beletal18}. Compared to de Bruijn graphs that are cyclic and have limited memory ($k$-mer length), our solution retains the linear structure of the block graph. 

\section{Definitions and basic tools\label{sect:definitions}}

\subsection{Strings}

We denote integer intervals by $[i..j]$. Let $\Sigma = \{1, \ldots, \sigma\}$ be an alphabet of size $|\Sigma| = \sigma$.
A \emph{string} $T[1..n]$ is a sequence of symbols from $\Sigma$, i.e. $T\in \Sigma^n$, where 
$\Sigma^n$ denotes the set of strings of length $n$ under the alphabet $\Sigma$.
A \emph{suffix} of string $T[1..n]$ is $T[i..n]$ for $1\leq i\leq n$. 
A \emph{prefix} of string $T[1..n]$ is $T[1..i]$ for $1\leq i\leq n$. 
A \emph{substring} of string $T[1..n]$ is $T[i..j]$ for $1\leq i\leq j \leq n$. 
Substring $T[i..j]$ where $j<i$ is defined as the \emph{empty string}.

The \emph{lexicographic order} of two strings $A$ and $B$ is naturally defined by the order of the alphabet: $A<B$ iff $A[1..i]=B[1..i]$ and $A[i+1]<B[i+1]$ for some $i\geq 0$. If $i+1>\min(|A|,|B|)$, then the shorter one is regarded as smaller. However, we usually avoid this implicit comparison by adding \emph{end marker} $\mathbf{0}$ to the strings.

Concatenation of strings $A$ and $B$ is denoted $AB$.

\subsection{Founder block graphs}

As mentioned in the introduction, our goal is to compactly represent a \emph{gapless} multiple sequence alignment (MSA) using a founder block graph. In this section we formalize these concepts.

A \emph{gapless multiple sequence alignment} $\mathtt{MSA}[1 \ltdots m,\, 1 \ltdots n]$ is a set of $m$ strings drawn from $\Sigma$, each of length $n$. Intuitively, it can be thought of as a matrix in which each row is one of the $m$ strings. 
Such a structure can be partitioned into what we call a \emph{segmentation}, that is, a collection of sets of shorter strings that can represent the original alignment.

\begin{definition}[Segmentation]
Let $\mathtt{MSA}[1 \ltdots m,\, 1 \ltdots n]$ be a gapless multiple alignment and let $R_1, R_2, \ldots, R_m$ be the strings in MSA. A \emph{segmentation} $S$ of MSA is a set of $b$ sets of strings $S^1, S^2, \ldots, S^b$ such that for each $1 \leq i \leq b$ there exist $j^{(i)}$ and $k^{(i)}$ such that for each $s \in S^i$, $s = R_t[j^{(i)}, k^{(i)}]$ for some $1 \leq t \leq m$. Furthermore, it holds $j^{(1)}= 1$, $k^{(b)} = m$, and $k^{(i)} = j^{(i+1)} - 1$ for all $1\leq i <b$, so that  $S^1, S^2, \ldots, S^b$ \emph{covers} the MSA. We call $w(S^i) = k^{(i)} - j^{(i)} +1$ the width of $S^i$. 
\end{definition}

A segmentation of a MSA can naturally lead to the construction of a founder block graph. Let us first introduce the definition of a block graph.

\begin{definition}[Block Graph]
A \emph{block graph} is a graph $G=(V,E,\ell)$ where $\ell: V \rightarrow \Sigma^+$ is a function that assigns a string label to every node and for which the following properties hold.
\begin{enumerate}
	\item $\{V^1, V^2, \ldots, V^b\}$ is a partition of $V$, that is, $V = V^1 \cup V^2 \cup \ldots \cup V^b$ and $V^i \cap V^j = \emptyset$ for all $i\neq j$;
	\item if $(v,w) \in E$ then $v \in V^i$ and $w \in V^{i+1}$ for some $1 \leq i \leq b-1$;
	\item if $v,w \in V^i$ then $|\ell(v)| = |\ell(w)|$ for each $1 \leq i \leq b$ and if $v\neq w$, $l(v) \neq l(w)$.
\end{enumerate}
\end{definition}
As a convention we call every $V^i$ a \emph{block} and every $\ell(v)$ a \emph{segment}.

Given a segmentation $S$ of a MSA, we can define the \emph{founder block graph} as a block graph \emph{induced} by $S$. The idea is to have a graph in which the nodes represents the strings in $S$ while the edges retain the information of how such strings can be recombined to spell any sequence in the original MSA.

\begin{definition}[Founder Block Graph]
A \emph{founder block graph} is a block graph $G(S) = (V,E,\ell)$ \emph{induced} by $S$ as follows: For each $1 \leq i \leq b$ we have $S^i = \{\ell(v) : v \in V^i\}$ and $(v,w) \in E$ if and only if there exists $i \in [1,b]$ and $t \in [1,m]$ such that $v \in V^i$, $w \in V^{i+1}$ and $R_{t}[j,j+|\ell(v)|+|\ell(w)|-1] = \ell(v)\ell(w)$ with  $j = 1+\sum_{h=1}^{i-1}w(S^h)$.
\end{definition}

We regard the edges of (founder) block graphs to be directed from left to right. Consider a path $P$ in $G(S)$ between any two nodes. The label $\ell(P)$ of $P$ is the concatenation of labels of the nodes in the path. Let $Q$ be a query string. We say that $Q$ \emph{occurs} in $G(S)$ if $Q$ is a substring of $\ell(P)$ for any path $P$ of $G(S)$.
Figure~\ref{fig:occurs} illustrates such queries.

\begin{figure}
    \centering
    \includegraphics{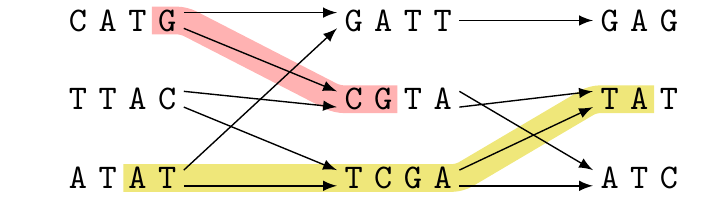}
    \caption{An example of two strings, \textcolor{red}{\texttt{GCG}} and \textcolor{yellow!70!blue}{\texttt{ATTCGATA}}, occurring in $G(S)$.}
    \label{fig:occurs}
\end{figure}

In our example in Figure~\ref{fig:founder-alignment}, segmentation $S$ would be $[1..4],[5..8],[9..11]$, and the induced founder block graph has thus 3 blocks with 9 nodes and 11 edges in total.

\subsection{Basic tools}

A \emph{trie} \cite{Bri59} of a set of strings is a rooted directed tree with outgoing edges of each node labeled by distinct characters such that there is a root to leaf path spelling each string in the set; shared part of the root to leaf paths to two different leaves spell the common prefix of the corresponding strings. 
Such a trie can be computed in $O(N \log \sigma)$ time, where $N$ is the total length of the strings, and it supports string queries that require $O(q \log \sigma)$ time, where $q$ is the length of the queried string. 

An \emph{Aho-Corasick automaton} \cite{AC75} is a trie of a set of strings with additional pointers (fail-links). While scanning a query string, these pointers (and some shortcut links on them) allow to identify all the positions in the query at which a match for any of the strings occurs. Construction of the automaton takes the same time as that of the trie. Queries take $O(q \log \sigma +\mathtt{occ})$ time, where $\mathtt{occ}$ is the number of matches.

A \emph{suffix array} \cite{MM93} of string $T$ is an array $\mathtt{SA}[1..n+1]$ such that $\mathtt{SA}[i]=j$ if $T'[j..n+1]$ is the $j$-th smallest suffix of string $T'=T\mathbf{0}$, where $T \in \{1,2,\ldots,\sigma\}$, and $\mathbf{0}$ is the end marker. Thus, $\mathtt{SA}[1]=n+1$. 

\emph{Burrows-Wheeler transform} BWT$[1..n+1]$ \cite{BW94} of string $T$ is such that BWT$[i]=T'[\mathtt{SA}[i]-1]$, where $T'=T\mathbf{0}$ and $T'[-1]$ is regarded as $T'[n+1]=\mathbf{0}$.

A \emph{bidirectional BWT index} \cite{SOG12, BCKM20} is a succinct index structure based on some auxiliary data structures on BWT. Given a string $T \in \Sigma^n$, with $\sigma\leq n$, such index occupying $O(n \log \sigma)$ bits of space can be build in $O(n)$ time and it supports finding in $O(q)$ time if a query string $Q[1..q]$ appears as substring of $T$ \cite{BCKM20}. Moreover, such query returns an interval pair  $([i..j]$,$[i'..j'])$ such that suffixes of $T$ starting at positions $\mathtt{SA}[i],\mathtt{SA}[i+1],\ldots, \mathtt{SA}[j]$ share a common prefix matching the query. Interval $[i'..j']$ is the corresponding interval in the suffix array of the reverse of $T$. Let $([i..j]$,$[i'..j'])$ be the interval corresponding to query substring $Q[l..r]$. A \emph{bidirectional backward step} updates the interval pair $([i..j]$,$[i'..j'])$ to the corresponding interval pair when the query substring $Q[l..r]$ is extended to the left into $Q[l-1..r]$ or to the right into $Q[l..r+1]$. This takes constant time \cite{BCKM20}.
A \emph{fully-functional bidirectional BWT index} \cite{BC19} expands the steps to allow contracting symbols from the left or from the right. That is, substring $Q[l..r]$ can be modified into $Q[l+1..r]$ or to $Q[l..r-1]$ and the the corresponding interval pair can be updated in constant time.

Among the auxiliary structures used in BWT-based indexes, we explicitly use the \emph{rank} and \emph{select} structures: String $B[1..n]$ from binary alphabet is called a \emph{bitvector}. Operation $\mathtt{rank}(B,i)$ returns the number of $1$s in $B[1..i]$. Operation $\mathtt{select}(B,j)$ returns the index $i$ containing the $j$-th $1$ in $B$.  Both queries can be answered in constant time using an index requiring $o(n)$ bits in addition to the bitvector itself \cite{Jac89}.

\section{Subclass of founder block graphs admitting indexing\label{sect:repeat-freeness}}

We now show that there exist a family of founder block graphs that admit a polynomial time constructable index structure supporting fast string matching. First, a trivial observation: the input multiple alignment is a founder block graph for the segmentation consisting of only one segment. Such founder block graph (set of sequences) can be indexed in linear time to support linear time string matching \cite{BCKM20}. Now, the question is, are there other segmentations that allow the resulting founder block graph to be indexed in polynomial time? We show that this is the case.

\begin{definition}
Founder block graph $G(S)$ is \emph{segment repeat-free} if each $\ell(v)$ for $v\in V$ occurs exactly once in $G(S)$.
\end{definition}

Our example graph (Fig.~\ref{fig:founder-alignment}) is not quite segment repeat-free, as \texttt{TAT} occurs also as substring of paths starting with \texttt{ATAT}. 

\begin{proposition}
Segment repeat-free founder block graphs can be indexed in polynomial time to support polynomial time string queries.
\label{prop:blockrepeatfree}
\end{proposition}

To prove the proposition, we construct such an index and show how queries can be answered efficiently.

Let $P(v)$ be the set of all paths starting from node $v$ and ending in a sink node. Let $P(v,i)$ be the set of \emph{suffix path labels} $\{\ell(L)[i..]\mid L \in P(v)\}$ for $1\leq i \leq |\ell(v)|$. Consider sorting $\mathcal{P}=\cup_{v \in V,1\leq i\leq |\ell(v)|} P(v,i)$ in lexicographic order. Then one can binary search any query string $Q$ in $\mathcal{P}$ to find out if it occurs in $G(S)$ or not. The problem with this approach is that $\mathcal{P}$ is of exponential size. 

However, if we know that $G(S)$ is segment repeat-free, we know that the lexicographic order of $\ell(L)[i..]$, $L \in P(v)$, is fully determined by the prefix $\ell(v)[i..|\ell(v)|]\ell(w)$ of $\ell(L)[i..]$, where $w$ is the node following $v$ on the path $L$. Let $P'(v,i)$ denote the set of suffix path labels cut at this manner. Now the corresponding set $\mathcal{P}'=\cup_{v \in V,1\leq i\leq |\ell(v)|} P'(v,i)$ is no longer of exponential size. Consider again binary searching a string $Q$ on sorted $\mathcal{P}'$. If $Q$ occurs in $\mathcal{P}'$ then it occurs in $G(S)$. If not, $Q$ has to have some $\ell(v)$ for $v\in V$ as its substring in order to occur in $G(S)$. 

To figure out if $Q$ contains $\ell(v)$ for some $v\in V$ as its substring, we build an Aho-Corasick automaton \cite{AC75} for $\{\ell(v) \mid v \in V\}$. Scanning this automaton takes $O(|Q|\log \sigma)$ time and returns such $v \in V$ if it exist. 

To verify such potential match, we need several tries \cite{Bri59}. For each $v\in V$, we build tries $\mathcal{R}(v)$ and $\mathcal{F}(v)$ on the sets 
$\{\ell(u)^{-1} \mid (u,v)\in E\}$ and $\{\ell(w) \mid (v,w)\in E\}$, respectively, where $X^{-1}$ denotes the reverse $x_{|X|}x_{|X|-1}\cdots x_1$ of string $X=x_1x_2 \cdots x_{|X|}$.

Assume now we have located (using Aho-Corasick automaton) $v\in V$ with $\ell(v)$ such that $\ell(v)=Q[i..j]$, where $v$ is at the $k$-th block of $G(S)$. We continue searching $Q[1..i-1]$ from right to left in trie $\mathcal{R}(v)$. If we reach a leaf after scanning $Q[i'..i-1]$, we continue the search with $Q[1..i'-1]$ on trie $\mathcal{R}(v')$, where $v'\in V$ is the node at block $k-1$ of $G(S)$ corresponding to the leaf we reached in the trie. If the search succeeds after reading $Q[1]$ we have found a path in $G(S)$ spelling $Q[1..j]$. We repeat the analogous procedure with $Q[j..m]$ starting from trie $\mathcal{F}(v)$.
That is, we can verify a candidate occurrence of $Q$ in $G(S)$ in $O(|Q|\log \sigma)$ time, as the search in the tries takes $O(\log \sigma)$ time per step. Note however, that there could be several labels $\ell(v)$ occurring as substrings of $Q$, so we need to do the verification process for each one of them separately. There can be at most $|Q|$ such candidate occurrences, due to the distinctness of node labels in $G(S)$. In total, this search can take at most $O(|Q|^2\log \sigma)$ time.

We are now ready to specify a theorem that reformulates Proposition~\ref{prop:blockrepeatfree} in detailed form.

\begin{theorem}
Let $G=(V,E)$ be a segment repeat-free founder block graph with blocks $V^1, V^2, \ldots, V^b$ such that $V=V^1 \cup V^2 \cup \cdots \cup V^b$. We can preprocess an index structure for $G$ in $O((N+W|E|)\log \sigma)$ time, where $\{1,\ldots,\sigma\}$ is the alphabet for node labels, $W=\max_{v \in V} \ell(v)$, and $N=\sum_{v \in V} \ell(v)$. Given a query string $Q[1..q] \in \{1,\ldots,\sigma\}^q$, we can use the index structure to find out if $Q$ occurs in $G$. This query takes $O(|Q|^2 \log \sigma)$ time. 
\label{thm:blockrepeatfree}
\end{theorem}
\begin{proof}
With preprocessing time $O(N\log \sigma)$ we can build the Aho-Corasick automaton \cite{AC75}. The tries can be built in $O(\log \sigma)(\sum_{v \in V} (\sum_{(u,v) \in E} |\ell(u)| + \sum_{(v,w) \in E} |\ell(w)|))= O(|E|W\log \sigma)$ time. The required queries on these structures take $O(|Q|^2 \log \sigma)$ time.

To avoid the costly binary search in sorted $\mathcal{P}'$, we instead construct the bidirectional BWT index \cite{BCKM20} for the concatenation $C=\prod_{i \in \{1,2,\ldots,b\}} \prod_{v \in V^i, (v,w) \in E} \ell(v)\ell(w)0$. Concatenation $C$ is thus a string of length $O(|E|W)$ from alphabet $\{\mathbf{0},1,2,\ldots, \sigma\}$. The bidirectional BWT index for $C$ can be constructed in $O(|C|)$ time, so that in $O(|Q|)$ time, one can find out if $Q$ occurs in $C$ \cite{BCKM20}. This query equals that of binary search in $\mathcal{P}'$.
\end{proof}

We remark that founder block graphs have a connection with \emph{generalized degenerate strings} (GD strings) \cite{Alz18}. In a GD string, sets of strings of equal length are placed one after the other to represent in a compact way a bigger set of strings. Such set contains all possible concatenations of those strings, which are obtained by scanning the individual sets from left to right and selecting one string from each set. The length of the strings in a specific set is called \emph{width}, and sum of all the width of all sets in a GD string is the \emph{total width}. Given two GD strings of the same total width it is possible to determine if the intersection of the sets of strings that they represent is non empty in linear time in the size of the GD strings \cite{Alz18}. Thus, the special case in which one of the two GD string is just a standard string can be seen also as a special case of a founder block graph in which every segment is fully connected with the next one and the length of the query string is equal to the maximal length of a path in the graph. 

We consider the question of indexing GD strings (fully connected block graphs) to search for queries $Q$ shorter than the total width. We can exploit the segment repeat-free property to yield such an index.

\begin{theorem}
Let $G=(V,E)$ be a \emph{segment repeat-free GD string} a.k.a. a fully connected segment repeat-free founder block graph with blocks $V^1, V^2, \ldots, V^b$ such that $V=V^1 \cup V^2 \cup \cdots \cup V^b$ and $(v,w)\in E$ for all $v\in V^i$ and $w \in V^{i+1}$, $1\leq i<b$. We can preprocess an index structure for $G$ in $O((N+W|E|)\log \sigma)$ time, where $\{1,\ldots,\sigma\}$ is the alphabet for node labels, $W=\max_{v \in V} \ell(v)$, and $N=\sum_{v \in V} \ell(v)$. Given a query string $Q[1..q] \in \{1,\ldots,\sigma\}^q$, we can use the index structure to find out if $Q$ occurs in $G$. This query takes $O(|Q| \log \sigma)$ time. 
\label{thm:blockrepeatfreeGDstring}
\end{theorem}
\begin{proof}
Recall the index structure of Theorem~\ref{thm:blockrepeatfree}. for the case of GD strings, we can simplify it as follow.  

We keep the same BWT index structure and the Aho-Corasick automaton, but we do not need any tries. After finding at most $|Q|$ occurrences of substrings of $Q$ in the graph using the Aho-Corasick automaton on node labels, we mark the matching blocks accordingly. If 2 marked blocks have exactly one marked neighboring block and $|Q|-2$ blocks have 2 marked neighboring blocks, then we have found an occurrence, otherwise not.  
\end{proof}

Observe that $\max(N,W|E|)\leq mn$, where $m$ and $n$ are the number of rows and number of columns, respectively, in the  multiple sequence alignment from where the founder block graph is induced. That is, the index construction algorithms of the above theorems can be seen to be almost linear time in the (original) input size. We study succinct variants of these indexes in Sect.~\ref{sect:compressedindexing}, and also improve the construction and query times to linear as side product.

\section{Construction of segment repeat-free founder block graphs \label{sect:construction}}

We can adapt the dynamic programming segmentation algorithms for founders \cite{NCKM19,CKMN19} to produce segment repeat-free founder block graphs.

The idea is as follows. Let $S$ be a segmentation of $\mathtt{MSA}[1..m,1..n]$. We say $S$ is \emph{valid} if it induces a segment repeat-free founder block graph $G(S)=(V,E)$. We build such valid $S$ for prefixes of MSA from left to right, minimizing the maximum block length needed.

\subsection{Characterization lemma}

Given a segmentation $S$ and founder block graph $G(S)=(V,E)$ induced by $S$, we can ensure that it is valid by checking if, for all $v\in V$, $\ell(v)$ occurs in the rows of the MSA only in the interval of the block $V^i$, where $V^i$ is the block of $V$ such that $v \in V^i$.

\begin{lemma}[Characterization]
\label{lemma:charaterization}
Let $j^{(i)} = 1+\sum_{h=1}^{i-1}w(S^h)$. A segmentation $S$ is valid if and only if, for all $V^i \subseteq V$, $1 \leq t \leq m$ and $j \neq j^{(i)}$, if $v \in V^i$ then $R_t[j,j+|\ell(v)|-1] \neq \ell(v)$.
\end{lemma}
\begin{proof}
To see that this is a necessary condition for the validity of $S$, notice that each row of MSA can be read through $G$, so if $\ell(v)$ occurs elsewhere than inside the block, then these extra occurrences make $S$ invalid. To see that this is a sufficient condition for the validity of $S$, we observe the following:
\begin{enumerate}[a)]
    \item For all $(v,w)\in E$, $\ell(v)\ell(w)$ is a substring of some row of the input MSA.
    \item Let $(x,u),(u,y)\in E$ be two edges such that $U=\ell(x)\ell(u)\ell(y)$ is not a substring of any row of input MSA. Then any substring of $U$ either occurs in some row of the input MSA or it includes $\ell(u)$ as its substring.
    \item If $\ell(v')$ for some $v'\in V$ does not occur as substring of some input MSA row, but it occurs as a substring of some path in $G$, then it includes $\ell(u)$ as its substring for some $u\in V$, where $|\ell(u)|<|\ell(v')|$. Such $\ell(u)$ occurs at least once outside the block $V^i$, where  $u\in V^i$. That is, if $\ell(v')$ makes $S$ invalid then also $\ell(u)$ makes it invalid. 
    \item Continuing case c) inductively, at some point there has to be $\ell(u')$ for some $u'\in V$ such that $\ell(u')$ is substring of $\ell(v)\ell(w)$ for some $(v,w) \in E$. By a), such $\ell(u')$ is also a substring of some row of the input MSA. That is, $v'$ of case c) cannot exist. 
\end{enumerate}
\end{proof}

\subsection{From characterization to a segmentation\label{sect:recurrence}}

Let $s(j')$ be the score of an optimal (minimum scoring) valid segmentation $S$ of prefix $\mathtt{MSA}[1..m,1..j']$, where the score is defined as $\max_{[a..b] \in S} b-a+1$.
We can compute
\begin{equation}
s(j)=\min_{j':1\leq j'\leq v(j)} \max(j-j',s(j')),
\label{eq:score}
\end{equation}
where $v(j)$ is the largest integer such that segment $\mathtt{MSA}[1..m,v(j)+1..j]$ is valid. The segment is valid iff each substring $\mathtt{MSA}[i,v(j)+1..j]$, for $1\leq i\leq m$, occurs as many times in $\mathtt{MSA}[1..m,v(j)+1..j]$ as in the whole MSA. If such $v(j)$ does not exist for some $j$, we set $v(j)=0$.

We can compute the score $s(n)$ of the optimal segmentation of MSA in $O(ns_{\mathtt{max}})$ time after preprocessing values $v(j)$ in $O(mns_{\mathtt{max}}\log \sigma)$ time, where $s_{\mathtt{max}}=\max_{j: v(j)>0} s(j)$. 
For the former, one can start comparing $\max(j-j',s(j'))$ from $j'=v(j)$ decreasing $j'$ by one, and then the value $j-j'$ grows bigger than $s(j')$ at latest after $s(j)-(j-v(j))\leq s(j)$ steps. For preprocessing, we build the bidirectional BWT index of the MSA in $O(mn)$ time \cite{BCKM20}. At column $j$, consider the trie containing the reverse of the rows of $M[1..m,1..j]$. Search the trie paths from the bidirectional BWT index until the number of leaves in each trie subtree equals the length of the corresponding BWT interval. Let $j'$ be the column closest to $j$ where this holds for all trie paths. Then one can set $v(j)=j'$. The $O(m(j-v(j))\log \sigma)$ time construction of the trie has to be repeated for each column. As $j-v(j)\leq s(j)$, the claimed preprocessing time follows.

\subsection{Faster preprocessing\label{sect:preprocessing}}

We can do the preprocessing in $O(mn)$ time.

\begin{theorem}
Given a multiple sequence alignment $\mathtt{MSA}[1 \ldots m, 1 \ldots n]$, values $v(j)$ for each $1 \leq j \leq n$ can be computed in $O(mn)$ time, where $v(j)$ is the largest integer such that segment $\mathtt{MSA}[1\ltdots m,v(j)+1\ltdots j]$ is valid.
\end{theorem}
\begin{proof}
Let us build the bidirectional BWT index \cite{BCKM20} of MSA rows concatenated into one long string. We will run several algorithms in synchronization over this BWT index, but we explain them first as if they would be run independently.

Algorithm 1 searches in parallel all rows from right to left advancing each by one position at a time. Let $k$ be the number parallel of steps done so far. It follows that we can maintain a bitvector $M$ that at $k$-th step stores $M[i]=1$ iff $BWT[i]$ is the $k$-th last symbol of some row.  

Algorithm 2 uses the \emph{variable length sliding window} approach of Belazzougui and Cunial \cite{BC19} to compute values $v(j)$. Let the first row of MSA be $T[1\ltdots n]$. Search $T[1\ltdots n]$ backwards in the fully-functional bidirectional BWT index \cite{BC19}. Stop the search at $T[j'+1\ltdots n]$ such that the corresponding BWT interval $[i'\ltdots i]$ contains only suffixes originating from column $j'+1$ of the MSA, that is, spelling $\mathtt{MSA}[a,j'+1\ltdots n]...$ in the concatenation, for some rows $a$. Set $v^1(n)=j'$. Contract $T[n]$ from the search string and modify BWT interval accordingly \cite{BC19}. Continue the search to find $T[j'+1\ltdots n-1]$ s.t. again the corresponding BWT interval $[i'\ltdots i]$ contains only suffixes originating from column $j'$. Update $v^1(n-1)=j'$. Continue like this throughout $T$. Repeat the process for all rows, to compute $v^2(j),v^3(j),\ldots,v^m(j)$ for all $j$. 
Set $v(j)=\min_i v^i(j)$ for all $j$. 

Let us call the instances of the Algorithm 2 run on the rest of the rows as Algorithms $3, 4, \ldots, m+1$. 

Let the current BWT interval in Algorithms $2$ to $m+1$ be $[j'+1\ltdots j]$. The problematic part in them is checking if the corresponding \emph{active} BWT intervals $[i'_a\ltdots i_a]$ for Algorithms $a \in \{2,3,\ldots, m+1\}$ contain only suffixes originating from column $j'+1$. To solve this, we run Algorithm 1 as well as Algorithms $2$ to $m+1$ in synchronization so that we are at the $k$-th step in Algorithm 1 when we are processing interval $[j'+1\ltdots j]$ in rest of the algorithms, for $k=n-j'$. In addition, we maintain bitvectors $B$ and $E$ such that $B[i'_a]=1$ and $E[i_a]=1$ for $a \in \{2,3,\ldots, m+1\}$. For each $M[i]$ that we set to 1 at step $k$ with $B[i]=0$ and $E[i]=0$, we check if $M[i-1]=1$ and $M[i+1]=1$. If and only if this check fails on any $i$, there is a suffix starting outside column $j'+1$. This follows from the fact that each suffix starting at column $j'+1$ must be contained in exactly one of the distinct intervals of the set $I=\{[i'_a\ltdots i_a]\}_{a\in 2,3 \ldots m+1}$. This is because $I$ cannot contain nested interval pairs as all strings in segment $[j'+1\ltdots j]$ of MSA are equal length, and thus their BWT intervals cannot overlap except if the intervals are exactly the same.  

Finally, the running time is $O(mn)$, since each extend-left and contract-right operations take constant time \cite{BC19}, and since the bitvectors are manipulated locally only on indexes that are maintained as variables during the execution.
\end{proof}

\subsection{Faster main algorithm}

Recall Eq.~(\ref{eq:score}). Before proceeding to the involved optimal solution, we give some insights by first improving the running time to logarithmic per entry. 

As it holds $v(1)\leq v(2) \leq \cdots \leq v(n)$, the range where the minimum is taken grows as $j$ grows.
Now, $[j'..j'+s(j')]$ can be seen as the \emph{effect range} of $s(j')$: for columns $j>j'+s(j')$ the maximum from the options is $j-j'$. Consider maintaining (key, value) pairs $(s(j')+j',s(j'))$ in a binary search tree (BST). When computing $s(j)$ we should have pairs $(s(j')+j',s(j'))$ for $1\leq j'\leq v(j)$ in BST. Value $s(j)$ can be computed by taking range minimum on BST values with keys in range $[j..\infty]$. Such query is easy to solve in $O(\log n)$ time. If there is nothing in the interval, $s(j)=j-v(j)$. Since this is semi-open interval on keys in range $[1 \ldots 2n]$, BST can be replaced by van Emde Boas tree to obtain $O(n\log \log n)$ time computation of all values \cite{GBT84}. Alternatively, we can remove elements from the BST once they no longer can be answers to queries, and we can get $O(n\log s_{\mathtt{max}})$ solution. To obtain better running time, we need to exploit more structural properties of the recurrence.

Cazaux et al. \cite{CKMN19} considered a similar recurrence and gave a linear time solution for it. In what follows we modify that technique to work with valid ranges.

For $j$ between $1$ and $n$, we define
\[x(j) = \max Argmin_{j' \in [1\ltdots v(j)]} \max(j-j',s(j'))\]

\begin{lemma}
    For any $j\in [1\ltdots n-1]$, we have $x(j) \leq x(j+1)$.
\end{lemma}

\begin{proof}
By the definition of $x(.)$, for any $j\in [1\ltdots n]$, we have for $j' \in [1\ltdots x(j)-1]$, $\max(j-j',s(j')) \geq \max(j-x(j),s(x(j)))$ and for $j' \in [x(j)+1 \ltdots v(j)]$, $\max(j-j',s(j')) > \max(j-x(j),s(x(j)))$.

We assume that there exists $j \in [1\ltdots n-1]$, such that $x(j+1) < x(j)$. In this case, $x(j+1) \in [1\ltdots x(j)-1]$ and we have $\max(j-x(j+1),s(x(j+1))) \geq \max(j-x(j),s(x(j)))$. As $v(j+1) \geq v(j)$, $x(j) \in [x(j+1)+1\ltdots v(j+1)]$ and thus $\max(j+1-x(j+1),s(x(j+1))) < \max(j+1-x(j),s(x(j)))$. As $x(j+1) < x(j)$, we have $j-x(j+1) > j-x(j)$.
To simplify the proof, we take $A = j-x(j+1)$, $B = s(x(j+1))$, $C= j - x(j)$ and $D = s(x(j))$. Hence, we have $\max(A,B) \geq \max(C,D)$, $\max(A+1,B) < \max(C+1,D)$ and $A > C$. Now we are going to prove that this system admits no solution.
\begin{itemize}
    \item Case where $A = \max(A,B)$ and $C = \max(C,D)$. As $A> C$, we have $A+1 > C+1$ and thus $\max(A+1,B) > \max(C+1,D)$ which is impossible because $\max(A+1,B) < \max(C+1,D)$.
    \item Case where $B = \max(A,B)$ and $C = \max(C,D)$. We can assume that $B>A$ (in the other case, we take $A = \max(A,B)$) and as $A > C$, we have $B > C+1$ and thus $\max(A+1,B) > \max(C+1,D)$ which is impossible because $\max(A+1,B) < \max(C+1,D)$.
    \item Case where $A = \max(A,B)$ and $D = \max(C,D)$. We have $A > D$ and $A > C$, thus $\max(A+1,B) > \max(C+1,D)$ which is impossible because $\max(A+1,B) < \max(C+1,D)$.
    \item Case where $B = \max(A,B)$ and $D = \max(C,D)$. We have $B \geq D$ and $A > C$, thus $\max(A+1,B) \geq \max(C+1,D)$ which is impossible because $\max(A+1,B) < \max(C+1,D)$.
\end{itemize}
\end{proof}

\begin{lemma}
    By initialising $s(1)$ to a threshold $K$, for any $j\in [1\ltdots n]$, we have $s(j) \leq \max(j,K)$.
\end{lemma}

\begin{proof}
    We are going to show by induction. The base case is obvious because $s(1) = K \leq \max(1,K)$.
    As $s(j)=\min_{j':1\leq j'\leq v(j)} \max(j-j',s(j'))$, by using induction, $s(j) \leq \min_{j':1\leq j'\leq v(j)} \max(j,K) \leq \max(j,K)$
\end{proof}

By using $K = n+1$, we can have that each $s(j)$ is in $O(n)$.

\begin{lemma}\label{le:j:star}
    Given $j* \in [x(j-1)+1\ltdots v(j)]$, we can compute in constant time if
    \[
    j^{\star} = \max Argmin_{j' \in [j^{\star} \ltdots v(j)]} \max(j-j',s(j')).
    \]
\end{lemma}

\begin{proof}
    We need just to compare $k = \max(j-j^{\star},s(j^{\star}))$ and $s(j^{\diamond})$ where $j^{\diamond}$ is in \\$Argmin_{j' \in [j^{\star}+1\ltdots v(j)]} s(j')$. If $k$ is smaller than $s(j^{\diamond})$, $k$ is smaller than all the $s(j')$ with $j' \in [j^{\star}+1\ltdots  v(j)]$ and thus for all $\max(j-j',s(j'))$. Hence we have \\$j^{\star}= \max Argmin_{j' \in [j^{\star} \ltdots v(j)]} \max(j-j',s(j'))$.

    Otherwise, $s(j^{\diamond}) \geq k$ and as $k \geq j-j^{\star}$,  $\max(j-j^{\diamond},s(j^{\diamond})) \geq k$. In this case $j^{\star} \neq \max Argmin_{j' \in [j^{\star}\ltdots  v(j)]} \max(j-j',s(j'))$.
    By using the constant time semi-dynamic range maximum query by Cazaux et al. \cite{CKMN19} on the array $s(.)$, we can obtain in constant time $j^{\diamond}$ and thus check the equality in constant time.
\end{proof}

\begin{theorem}
    We can build all the values $s(j)$ in $O(n)$ time after $O(nm)$ preprocessing.
\end{theorem}

\begin{proof}
    After preprocessing in $O(mn)$ to compute all the values $v(j)$, we can build all the values $s(j)$ by increasing $j$ and computing the values $x(j)$. For any $j \in [1\ltdots n]$, we check any $j'$ from $x(j-1)$ to $v(j)$ with the equality of Lemma~\ref{le:j:star} until one is true. With the property of the equality, we know that this $j'$ is $x(j)$. After that we compute $s(j)$ which corresponds to the value $\max(j-x(j),s(x(j)))$, we add $s(j)$ to the constant time semi-dynamic range maximum query and we compute $s(j+1)$.
\end{proof}

\section{Succinct index for segment-free founder block graphs \label{sect:compressedindexing}}

Recall the indexing solutions of Sect.~\ref{sect:repeat-freeness} and the definitions from Sect.~\ref{sect:definitions}. 

We now show that explicit tries and Aho-Corasick automaton can be replaced by some auxiliary data structures associated with the Burrows-Wheeler transformation of the concatenation $C=\prod_{i \in \{1,2,\ldots,b\}} \prod_{v \in V^i, (v,w) \in E} \ell(v)\ell(w)\mathbf{0}$. 

Consider interval $\mathtt{SA}[i\ltdots k]$ in the suffix array of $C$ corresponding to suffixes having $\ell(v)$ as prefix for some $v \in V$. From the segment repeat-free property it follows that this interval can be split into two subintervals, $\mathtt{SA}[i\ltdots j]$ and $\mathtt{SA}[j+1\ltdots k]$, such that suffixes in $\mathtt{SA}[i\ltdots j]$ start with $\ell(v)\mathbf{0}$ and suffixes in $\mathtt{SA}[j+1\ltdots k]$ start with $\ell(v)\ell(w)$, where $(v,w) \in E$. Moreover, BWT$[i\ltdots j]$ equals multiset $\{\ell(u)[|\ell(u)|-1] \mid (u,v) \in E\}$ \emph{sorted in lexicographic order}. This follows by considering the lexicographic order of suffixes $\ell(u)[|\ell(u)|-1]\ell(v)\mathbf{0}\ldots$ for $(u,v)\in E$. That is, BWT$[i\ltdots j]$ (interpreted as a set) represents the children of the root of the trie $\mathcal{R}(v)$ considered in Sect.~\ref{sect:repeat-freeness}.

We are now ready to present the search algorithm that uses only the BWT of $C$ and some small auxiliary data structures. We associate two bitvectors $B$ and $E$ to the BWT of $C$ as follows. We set $B[i]=1$ and $E[k]=1$ iff $\mathtt{SA}[i\ltdots k]$ is maximal interval with all suffixes starting with $\ell(v)$ for some $v\in V$.

Consider the backward search with query $Q[1\ltdots q]$. Let $\mathtt{SA}[j'\ltdots k']$ be the interval after matching the shortest suffix $Q[q'\ltdots q]$ such that BWT$[j']=\mathbf{0}$. Let $i=\mathtt{select}(B,\mathtt{rank}(B,j'))$ and $k=\mathtt{select}(E,\mathtt{rank}(B,j'))$.
If $i\leq j'$ and $k'\leq k$, index $j'$ lies inside an interval $\mathtt{SA}[i\ltdots k]$ where all suffixes start with $\ell(v)$ for some $v$. We modify the range into $\mathtt{SA}[i\ltdots k]$, and continue with the backward step on $Q[q'-1]$. We check the same condition in each step and expand the interval if the condition is met. Let us call this procedure \emph{expanded backward search}.

We can now strictly improve Theorems~\ref{thm:blockrepeatfree}~and~\ref{thm:blockrepeatfreeGDstring} as follows.

\begin{theorem}
Let $G=(V,E)$ be a segment repeat-free founder block graph (or a segment repeat-free GD string) with blocks $V^1, V^2, \ldots, V^b$ such that $V=V^1 \cup V^2 \cup \cdots \cup V^b$. We can preprocess an index structure for $G$ occupying $O(W|E|\log\sigma)$ bits in $O(W|E|)$ time, where $\{1,\ldots,\sigma\}$ is the alphabet for node labels and $W=\max_{v \in V} \ell(v)$. Given a query string $Q[1\ltdots q] \in \{1,\ldots,\sigma\}^q$, we can use expanded backward search with the index structure to find out if $Q$ occurs in $G$. This query takes $O(|Q|)$ time. 
\label{thm:blockrepeatfreeBWTindex}
\end{theorem}
\begin{proof} (sketch)
As we  expand the search interval in BWT, it is evident that we still find all occurrences for short patterns that span at most two nodes, like in the proof of Theorem~\ref{thm:blockrepeatfree}. We need to show that a) the expansions do not yield spurious occurrences for such short patterns and b) the expansions yield exactly the occurrences for long patterns that we earlier found with the Aho-Corasick and tries approach. 

In case b), notice that after an expansion step we are indeed in an interval $\mathtt{SA}[i\ltdots k]$ where all suffixes match $\ell(v)$ and thus corresponds to a node $v \in V$. The suffix of the query processed before reaching interval $\mathtt{SA}[i\ltdots k]$ must be at least of length $|\ell(v)|$. That is, to mimic Aho-Corasick approach, we should continue with the trie $\mathcal{R}(v)$. This is identical to taking backward step from BWT$[i\ltdots k]$, and continuing therein to follow the rest of this implicit trie. 

To conclude case b), we still need to show that we reach all the same nodes as when using Aho-Corasick, and that the search to other direction with $\mathcal{L}(v)$ can be avoided. These follow from case a), as we see.

In case a), before doing the first expansion, the search is identical to the original algorithm in the proof of Theorem~\ref{thm:blockrepeatfree}. After the expansion, all matches to be found are those of case b). That is, no spurious matches are reported. Finally, no search interval can include two distinct node labels, so the search reaches the only relevant node label, where the Aho-Corasick and trie search simulation takes place. We reach all such nodes that can yield a full match for the query.
\end{proof}

\section{Gaps in multiple alignment\label{sect:gaps}}

We have so far assumed that our input is a gapless multiple alignment. Let us now consider how to extend the results to the general case.
The idea is that gaps are only used in the segmentation algorithm to define the valid ranges, and that is the only place where special attention needs to be taken; elsewhere, whenever a substring from MSA rows is read, gaps are treated as empty strings. That is, A-GC-TA- becomes AGCTA.  

It turns out that allowing gaps in MSA indeed makes the computation of valid ranges more difficult. To see this,
consider a segment in MSA containing strings:
\begin{verbatim}
-AC-CGATC-
-A-CCGATCC
AAC-CGATC-
AAC-CGA-C-
\end{verbatim}
After gaps are removed this segment becomes:
\begin{verbatim}
ACCGATC
ACCGATCC
AACCGATC
AACCGAC
\end{verbatim}
Without even seeing the rest of the MSA, one can see that this is not a valid block, as the first string is a prefix of the second. With gapless MSAs this was not possible and the algorithm in Sect.~\ref{sect:preprocessing} exploited this fact. Modifying that algorithm to handle gaps properly is possible, but non-trivial. 

We leave this extension to future work; see however Sections \ref{sect:experiments} and \ref{sect:discussion} for some further insights. 

Despite the preprocessing for the segmentation is affected by gaps in MSA, once such valid segmentation is found, the rest of the results stay unaffected.
All the proposed definitions extend to this interpretation by just omitting gap symbols when reading the strings. The consequence for founder block graph is that strings inside a block can be variable length. Interestingly, with this interpretation Theorem~\ref{thm:blockrepeatfreeGDstring} can be expressed with GD strings replaced by \emph{elastic strings} \cite{bernardini_et_al2019elastic}. 

\section{Implementation and experiments \label{sect:experiments}}

We implemented several methods proposed in this paper. The implementation is available at \url{https://github.com/algbio/founderblockgraphs}. Some preliminary experiments are reported below. 

\subsection{Construction}

We implemented the founder block graph construction algorithm of Sect.~\ref{sect:recurrence} with the faster preprocessing routine of Sect.~\ref{sect:preprocessing}; in place of fully-functional bidirectional BWT index, we used similar routines implemented in compressed suffix trees of SDSL library \cite{sdsl}. 

To test the implementation we downloaded 1484 strains of SARS-CoV-2 strains stored in NCBI database.\footnote{https://www.ncbi.nlm.nih.gov/, accessed 24.04.2020.} 
We created a multiple sequence alignment of the strains using \emph{ViralMSA}\cite{viralMSA}. We filtered out rows that contained gaps or N's. 
We were left with a multiple sequence alignment of 410 rows and 29811 columns. Our algorithm took 58 seconds to produce the optimal segmentation on Intel(R) Xeon(R) CPU, E5-2690, v4, 2.60GHz. There were 3352 segments in the segmentation, the maximum segment length was 12, and the maximum number of founder segments in a block was 12. The founder block graph had 3900 nodes and 4440 edges. The total length of node labels was 34968. The graph size is thus less than $1\%$ of the MSA size.

We also implemented support to construct founder block graphs for general MSA's that contain gaps. The modification to the gapless case was that nested BWT intervals needed to be detected. We stopped left-extension as soon as the BWT intervals contained no other repeats than those caused by nestedness. This left valid range undefined on such columns, but for the rest the valid range can still be computed correctly (undefined values could be postprocessed using matching statistics \cite{BCD18} on all pairs of prefixes preceding suffixes causing nested intervals). Initial experiments show very similar behaviour to the gapless case, but we defer further experiments until the implementation is mature enough.

\subsection{Indexing}

We implemented the succinct indexing approach of Sect.~\ref{sect:compressedindexing}. On the founder block graph of the previous experiment, the index occupied 87 KB. This is $3\%$ of the original input size, as the encoding of the input MSA with 2 bits per nucleotide takes 2984 KB. 

Figure~\ref{fig:plot} shows an experiment with indexes built on different size samples of the MSA rows, and with querying patterns of varying length sampled from the same rows. As can be seen, the query times are not affected by the size of the MSA samples (showing independence of the input MSA), but only on their length (showing linear dependency on the query length).  

\begin{figure}
    \centering
    \includegraphics[width=0.9\textwidth]{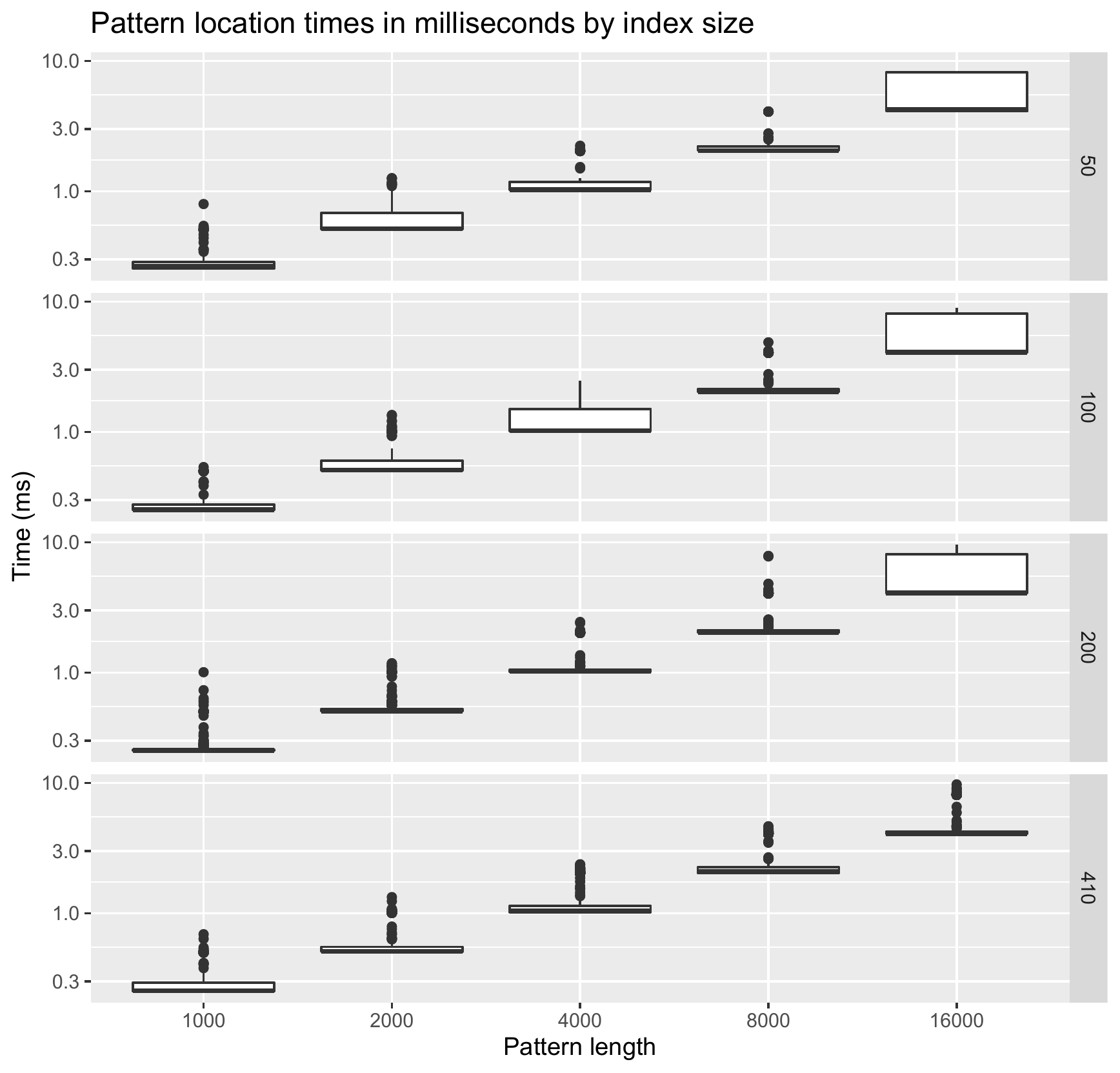}
    \caption{Running time of querying patterns from the founder block graph. The sample sizes (for MSA row subsets) are shown on the right-hand side of each plot. The plots show averages and distribution over 10 repeats of each search, where one search consist of a set of query patterns of given length randomly sampled from the respective MSA row subset. The pattern set sample size (10,20,30,40, respectively) grows by the MSA sample size, but the reported numbers are normalized so that the query time (milliseconds) is per pattern. This experiment was run on Intel(R) Core(TM) i5-4308U CPU, 2.80GHz.\label{fig:plot}}
\end{figure}

\section{Discussion\label{sect:discussion}}

One characterization of our solution is that we compact those vertical repeats in MSA that are not horizontal repeats. This can be seen as positional extension of variable order de Bruijn graphs. Also, our solution is parameter-free unlike de Bruijn approaches that always need some threshold $k$, even in the variable order case. 

The founder block graph concept could also be generalized so that it is not directly induced from a segmentation. One could consider cyclic graphs having the same segment repeat-free property. This could be interesting direction in defining parameter-free de Bruijn graphs.

Future work include a proper extension of the approach to general MSA's containing gaps. Our implementation already contains support for such MSA's, but the theory framework still requires some more work to show that such extension can be done without any effect on the running time. 
 
On the experimental side, there is still much more work to be done. So far, we have not optimized any of the algorithms for multithreading nor for space usage. For example, one could use BWT indexes engineered for highly repetitive text collections to build the founder block graphs in space proportional of the compressed MSA. Such optimizations are essential for applying the approach on e.g. human genome data. Past experience on similar solutions \cite{NCKM19} indicate that our approach should easily be applicable to much larger datasets than those we covered in our preliminary experiments.

Finally, this paper only scratches the surface of a new family of pangenome representations. There are myriad of options how to optimize among the valid segmentations \cite{NCKM19,CKMN19}, e.g. by optimizing the number of founder segments \cite{NCKM19} or controlling the over-expressiveness of the graph, rather than minimizing the maximum block size as studied here.

\end{document}